\documentclass[conference]{IEEEtran}
\IEEEoverridecommandlockouts
\usepackage{cite}
\usepackage{amsmath,amssymb,amsfonts}
\usepackage{amsthm}
\usepackage{algorithmic}

\usepackage{graphicx}
\usepackage{textcomp}
\usepackage{xcolor}
\def\BibTeX{{\rm B\kern-.05em{\sc i\kern-.025em b}\kern-.08em
    T\kern-.1667em\lower.7ex\hbox{E}\kern-.125emX}}

\usepackage{fancyhdr}
\newtheorem{definition}{Definition}
\newtheorem{proposition}{Proposition}
\begin{document}

\title{The Ordinal Annotation Game: How Construct Abstraction Shapes Crowdsourced Consensus}

\author{\IEEEauthorblockN{Kosmas Pinitas\\University of Piraeus\\kpinitas@unipi.gr}}

\maketitle
\thispagestyle{fancy}

\begin{abstract}
Inter-annotator disagreement in real-time affect annotation is widely treated as stochastic noise. We challenge this view by modelling ordinal annotation as an implicit game-theoretic coordination process against an internalised population prior under a post-hoc majority vote. We present the Ordinal Annotation Game, a conceptual scaffold in which the mapping from individual effort to collective consensus is governed by the semantic abstraction of the target construct. We evaluate it across two experiments sharing identical interface software and a uniform sensitivity threshold: a controlled sensory tracking study and an in-the-wild engagement study.  Sensory annotation yields a consensus-dominant regime where active updates reinforce agreement, whereas engagement annotation inverts into an effort-limited regime where more labelling penalises consensus. The payoff slope reverses sign under identical processing, showing that ordinal disagreement is a structured behavioural phenomenon, not a discretisation artefact or random error.
\end{abstract}


\section{Introduction}
Affect annotation is central to building responsive affective computing systems. Ordinal protocols have gained traction due to the fact that they alleviate the cognitive overhead and calibration issues of absolute-value tracking. However, considerable inter-annotator disagreement persists over time. Conventionally this disagreement is filtered, smoothed, or voted away as random noise \cite{pinitas2025privileged,makantasis2023lab}. Yet this overlooks the strategic nature of annotation: even without peer feedback, annotators balance the effort of altering an interface state against an implicit desire to align with an internalised population prior.

This paper formalises this dynamic by framing independent annotation as a repeated game of implicit strategic coordination \cite{camerer2003behavioral,weidenholzer2010coordination}. We introduce the \emph{Ordinal Annotation Game}, a conceptual scaffold in which annotators transform continuous perceptions into input curves that the designer operationalises into discrete directional actions under a just-noticeable-difference deadband ($\delta$). By keeping this threshold strictly uniform across investigations, we isolate a core factor driving crowdsourced divergence: \textit{construct abstraction}. For low-level perceptual attributes (e.g., sound pitch or colour intensity), a shared decision threshold drives the population toward a consensus-dominant equilibrium. In contrast, for complex psychological constructs (e.g., gameplay engagement), variability in individuals' internal appraisal criteria introduces private noise, shifting the system into an effort-limited regime in which additional individual effort increasingly decouples the judgments of an annotator from the group consensus.

\section{Related Work}\label{sec:related}
Continuous affect annotation has evolved from absolute coordinate tracking \cite{cowie2000feeltrace} toward rank-based, unbounded, and relative interfaces such as PAGAN \cite{melhart2019pagan} and RankTrace \cite{lopes2017ranktrace} to minimise annotator reaction lag and cognitive friction \cite{mariooryad2013analysis}. Pairwise, listwise, and ordinal methods \cite{martinez2014don,pinitas2024varying,pinitas2026lasca} leverage these inputs to capture affective trends. Whereas standard processing treats thresholds as nuisance parameters, we treat thresholding as a primary strategic behaviour. Moreover, traditional reliability coefficients \cite{cohen1960coefficient} and fusion frameworks \cite{ringeval2015av+} assume a singular, recoverable ground truth, whereas modern probabilistic crowd-modelling \cite{dawid1979maximum}, deep crowd layers \cite{rodrigues2018deep}, and distribution-learning paradigms \cite{geng2016label,pinitas2026beyond} preserve annotation ambiguity as a reflection of subjective variance. Finally, we move beyond statistical intersubject correlation \cite{hasson}, autonomic coupling \cite{palumbo2017interpersonal}, and recurrence analyses \cite{wallot2018analyzing} by reframing them within mechanism design and peer-prediction theory \cite{mohammad2013crowdsourcing}. This synthesis reveals construct abstraction as the critical control parameter underlying the transition from cooperative tracking to structural disagreement.

\section{The Ordinal Annotation Game Framework}
We formalise ordinal annotation as an implicit repeated \emph{coordination game} \cite{camerer2003behavioral,weidenholzer2010coordination} whose players are the independent annotators $\mathcal{N}$. Their actions are the discrete directional labels $\mathcal{A}$, and the payoff is the utility of Eq.~\eqref{eq:utility}. Annotators receive no real-time peer feedback, so they cannot best-respond to observed choices; instead they coordinate against an \emph{internalised} expectation of the group. Coordinating on such an internal focal point in the absence of communication is a well-established mechanism in the coordination-game and peer-prediction literature \cite{schelling1958strategy,prelec2004bayesian,miller2005eliciting}. Under this lens, Definition \ref{def:oag} specifies the game and Proposition \ref{prop:frontier} predicts how its equilibrium shifts with construct abstraction.

\subsection{Action Space and Threshold Mapping}\label{sec:setup}
Let $\mathcal{N}=\{1,\dots,N\}$ be the independent annotators tracking a stimulus over time $t$; the stimulus has an underlying true trace $\theta_t$ and observable physical saliency $s_t \ge 0$. At each step, annotator $i$ experiences an internal perceptual change signal $\Delta\hat{z}_{i,t}$---their subjective evaluation of the stimulus delta---and records a continuous trajectory change. To filter micro-fluctuations, the designer applies a post-hoc transformation bounded by a sensitivity deadband of half-width $\delta > 0$ (Fig. \ref{fig:threshold}), mapping the continuous user traces onto a three-element ordered action space:
\begin{equation}
  a_{i,t}\in\mathcal{A}=\{-1,\,0,\,+1\}.
  \label{eq:action}
\end{equation}
The analytical population consensus is defined post-hoc via a majority aggregation rule $y_t = \Pi(a_t) = \text{mode}(a_{1,t}, \dots, a_{N,t})$, the \emph{instantaneous mode}, i.e.\ the most frequent action across annotators within a single time window $t$. This aggregate is computed by the analyst after data collection and is \emph{never} shown to annotators during labelling; the coordination it captures is therefore against each annotator's internalised expectation of the group, not against observed peer labels.

\begin{figure}[t]
\centerline{\includegraphics[width=0.5\columnwidth]{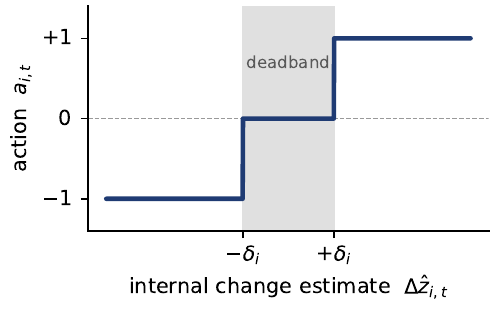}}
\caption{The designer's post-hoc mapping from continuous perception to
discrete ordinal labels. Structural disagreement
emerges as the construct shifts from concrete physical properties to
abstract psychological criteria.}
\label{fig:threshold}
\end{figure}

\subsection{Utility as a Conceptual Scaffold}\label{sec:eq}
We express each annotator's latent instantaneous utility at time $t$ as:
\begin{equation}
  u_{i,t} = \alpha(s_t)\,\mathbb{E}_{y_t}\!\big[C_{i,t}\big] - \beta(s_t)\,E_{i,t} - \gamma\,B_{i,t},
  \label{eq:utility}
\end{equation}
where individual coordination $C_{i,t} = \mathbb{I}(a_{i,t} = y_t)$ captures alignment relative to the expected consensus. Behavioural effort tracks label volatility through the \emph{unsigned} step magnitude, $E_{i,t}=|a_{i,t}-a_{i,t-1}|$, whereas the bias term penalises systematic directional drift through the \emph{signed} counterpart, $B_{i,t}=\big|\sum_{\tau\le t}(a_{i,\tau}-a_{i,\tau-1})\big|=|a_{i,t}-a_{i,0}|$, computed on the same discretised action stream.  $B_{i,t}$ is the net rather than the total displacement and consequently is bounded by cumulative effort ($B_{i,t}\le\sum_{\tau\le t}|a_{i,\tau}-a_{i,\tau-1}|$) yet remains distinct from it: an annotator who updates frequently but symmetrically incurs high effort and low bias. We stress that $B_{i,t}$ appears only in this conceptual utility; it is \emph{not} part of the empirical analysis, which relies solely on effort $E_i$ and agreement $A_i$ (Section~\ref{sec:pred}). We do not treat Eq.~\eqref{eq:utility} as a fitted predictive model or estimate $\alpha,\beta,\gamma$: it is a conceptual vehicle formalising why independent actions sample a stable distribution $P^{\star}(a \mid s_t, \Pi, \delta, \mathbf{\Phi})$, where $\mathbf{\Phi}$ captures the shared or unshared semantic parameters of the construct, motivating the empirical contrast rather than predicting affect trajectories.

\begin{definition}[Ordinal Annotation Game]
\label{def:oag}
Given an aggregation rule $\Pi$, an independent population $\mathcal{N}$ analysed under an explicit deadband $\delta$, and a stimulus path $(\theta_t,s_t)$, the Ordinal Annotation Game is the repeated process wherein independent annotators generate continuous evaluations $\Delta\hat{z}_{i,t}$ that optimise their internal utility expectation \eqref{eq:utility}; the resulting time-series values manifest as empirical samples from the distribution $P^{\star}(a \mid s_t, \Pi, \delta, \mathbf{\Phi})$.
\end{definition}

\subsection{The Unified Phase Transition Hypothesis}\label{sec:pred}
Our formulation explores the empirical relationship between an annotator's labelling effort ($E_i$) and their aggregate consensus agreement ($A_i$). We define these windowed metrics over an operational analytical interval $\mathcal{T}$ containing $T$ time steps:
\begin{equation}
  E_i = \frac{1}{T}\sum_{t \in \mathcal{T}} |a_{i,t} - a_{i,t-1}|,
  \qquad
  A_i = \frac{1}{T}\sum_{t \in \mathcal{T}} \mathbb{I}(a_{i,t} = y_t).
  \label{eq:effort_agree_def}
\end{equation}
We estimate the payoff-frontier slope $dA/dE$ by ordinary least squares over the cohort: $A_i = m E_i + c$, with $m = dA/dE$.  $E_i$ and $A_i$ both derive from the same stream $\{a_{i,t}\}$ and hence they are \emph{not} independent---absent a shared anchor, a more volatile annotator is mechanically more likely to fall off the mode---so the raw sign of $dA/dE$ is necessary but not sufficient. Section~\ref{sec:residual} addresses this using the sensory condition as an identically-processed built-in control.

\begin{proposition}[Semantic Inversion of the Payoff Frontier]
\label{prop:frontier}
Let $\delta$ be fixed uniformly across all annotators. If the task construct is concrete and sensory, the internal change signal is tightly coupled to the physical trace ($\Delta\hat{z}_{i,t} = \Delta\theta_t + \varepsilon_{i,t}$, with minimal noise variance $\sigma_\varepsilon^2$), yielding a positive frontier slope ($dA/dE > 0$). If the task construct is abstract and psychological, internal changes incorporate private appraisal criteria ($\Delta\hat{z}_{i,t} = \Delta\theta_t + \nu_i(t) + \varepsilon_{i,t}$, expanding private variance $\sigma_\nu^2$), driving the slope negative ($dA/dE < 0$).
\end{proposition}

\begin{proof}[Sketch]
For concrete features, low $\sigma_\varepsilon^2$ means individual effort ($E_i$) reflects alertness along a shared physical signal, so high-effort annotators capture synchronised transitions that others agree with, giving positive effort--consensus covariance ($dA/dE > 0$). For abstract constructs, semantic appraisals ($\nu_i(t)$) diverge: active annotators trigger changes driven by private interpretations that fail to win the majority $y_t$, so effort executes uncoordinated private movements and the slope turns negative ($dA/dE < 0$). Since the mechanical co-determination of $E_i$ and $A_i$ is identical across conditions, the sign reversal (Section~\ref{sec:residual}) isolates the behavioural effect.
\end{proof}

\section{Experimental Design and Datasets}\label{sec:datasets}

To evaluate this semantic phase transition, we analyse data from two distinct ordinal annotation experiments that share the exact same interface software and physical threshold setting ($\delta = 0.05$). Both datasets are segmented using overlapping $3$\,s windows with a $1$\,s step size. All annotation was collected at the University of Malta. $N_1$ and $N_2$ denote two \emph{separate} annotator pools for two different tasks; they are not successive counts of a single sample. In both experiments annotators produced \emph{unbounded} continuous traces in PAGAN using the mouse wheel---scrolling upward or downward to raise or lower the tracked value in real time \cite{lopes2017ranktrace,melhart2019pagan}---which the designer subsequently discretised into the action space $\mathcal{A}$ through the uniform deadband ($\delta = 0.05$). Before thresholding, the traces were min--max normalised.

\subsection{Experiment 1: Controlled Sensory Tracking}
This experiment evaluates low-level sensory perception using a controlled change-detection design \cite{barthet2023knowing} implemented in the PAGAN interface \cite{melhart2019pagan} (Figures \ref{fig:tasks}.a, \ref{fig:tasks}.b). A population of $N_1 = 20$ independent annotators---the $5$ researchers who also took part in Experiment~2 together with $15$ additional student volunteers---tracked explicit, physical properties of continuous stimuli: \textbf{Colour Intensity Task (Visual):} Annotators report real-time directional changes in the colour luminance field of a visual canvas. \textbf{Sound Pitch Task (Auditory):} Annotators report real-time directional pitch changes  of the audio tone.

\subsection{Experiment 2: Abstract Psychological Tracking}
This setup leverages the GameVibe corpus \cite{gamevibe,barthet2023knowing} to evaluate high-level psychological appraisal. A cohort of $N_2 = 5$ researchers---the same five who also participated in Experiment~1---watched one-minute video clips captured from 30 commercial first-person shooter games (Fig. \ref{fig:tasks}.c), annotating the perceived gameplay engagement in real time using PAGAN \cite{melhart2019pagan}. Engagement was chosen as abstract construct since it is a canonical high-level, multi-cue psychological target in affective computing: it has no single physical referent and must be inferred by integrating many latent cues (e.g.\ player movement, combat intensity, interface events), placing it at the opposite end of the abstraction axis from the sensory primitives of Experiment~1 while remaining easy-to-annotate under the  PAGAN interface.

\subsection{Exogenous Saliency and Metric Modelling}

Saliency is defined per experiment. In Experiment~1 each stimulus
exposes a single primitive, so saliency is its rate of change:
luminance change (colour task) and pitch change (sound task). In
Experiment~2 the naturalistic gameplay has no single referent, so
audiovisual saliency comes from a unified pipeline capturing low-level
perceptual change from the raw stream
\cite{itti2001computational,kayser2005mechanisms}. Visual saliency
fuses frame-difference energy (L2 greyscale), optical-flow
magnitude (Farneb\"ack motion), and HSV histogram jump (L1 change,
$16\times16$ histograms), each smoothed and weighted-averaged. Audio
saliency applies the same processing to RMS energy, spectral flux, and
onset strength at $100$\,Hz. All channels are resampled to $10$\,Hz,
aligned to the shortest modality, normalised per
clip, and averaged within each window. The resulting trace captures
complementary primitives that drive attention and shared human reactions, constituting a
stimulus-driven signal that plausibly shapes annotator coordination.

To quantify stimulus-driven behaviour, low-level physical saliency curves ($s$) and group response entropy tracking ($H$) are computed within each $3$\,s window. Crucially, to align with the discrete action space $\mathcal{A}$, both the exogenous saliency fluctuations and response entropy steps are passed through the exact same sensitivity deadband transformation ($\delta = 0.05$), mapping all metrics onto a uniform categorical domain: $s_t, H_t, C_{i,t} \in \{-1,0,+1\}$.

Dynamic couplings are evaluated with Cohen's Kappa rather than a signed correlation due to the fact that both compared tracks already live on the same alphabet $\{-1,0,+1\}$, where a Pearson correlation over the sparse $\pm 1$ codes would overstate agreement. Crucially, $\kappa$ is a \emph{coupling diagnostic}: it does not assert that the two underlying constructs are the same quantity, but measures how often their discretised \emph{directions of change} co-occur above chance. We write $\rho(s,C)$ for the Saliency--Coordination Coupling (stimulus jumps vs.\ group coordination) and $\rho(s,H)$ for the Saliency--Entropy Coupling (salience vs.\ collective answer entropy), with time-resolved versions $\rho_t(\cdot,\cdot)$. High $H$ denotes strong disagreement, with choices scattered across the discrete space.

\subsubsection*{Methodological Limitations and Confounds}
While a uniform threshold $\delta = 0.05$ removes deadband discrepancies, a causal claim about construct abstraction remains confounded by differences between the corpora: Experiment~1 uses synthetic, low-complexity stimuli annotated by a general cohort, whereas Experiment~2 uses complex naturalistic gameplay annotated by expert researchers. These differences in stimulus complexity, population, and designed ground truth are unavoidable in secondary analysis, so we frame the contrast as a strong association rather than an isolated causal coupling.

\section{Equilibrium Analysis}\label{sec:results}

\subsection{The Consensus-Dominant Regime (Sensory Tracking)}
For concrete physical changes under the uniform threshold ($\delta=0.05$), the results confirm a consensus-dominant equilibrium: the population frontier slope $dA/dE$ is strongly positive for both colour intensity ($0.63$, $95\%$ bootstrap CI $[0.48, 0.76]$) and sound pitch ($0.57$, $[0.41, 0.71]$). This confirms the first branch of Proposition \ref{prop:frontier}---for tasks tied to sensory primitives, higher effort yields higher consensus. The couplings differ by modality: colour intensity tracks independently ($\rho(s,C) = -0.02$, $\rho(s,H) = 0.03$), while sound pitch shows an explicit stimulus anchor ($\rho(s,C) = 0.27$, $\rho(s,H) = 0.07$). The temporal micro-dynamics in Fig. \ref{fig:controlled_dyn} confirm this stable landscape: $dA/dE_t$ stays positive across nearly the entire horizon for colour intensity (peaking $\approx 1.8$ near $10$\,s), and for sound pitch it remains safely positive ($\approx 0.5$) even during the low-saliency window ($10$--$23$\,s) where $\rho_t(s,C)$ and $\rho_t(s,H)$ fall to $\approx -0.7$. A uniform discretisation thus sustains a cooperative equilibrium even when stimulus anchors fade.

\begin{figure}[t]
\centering
\begin{minipage}{0.3\columnwidth}\centering
\includegraphics[width=\linewidth]{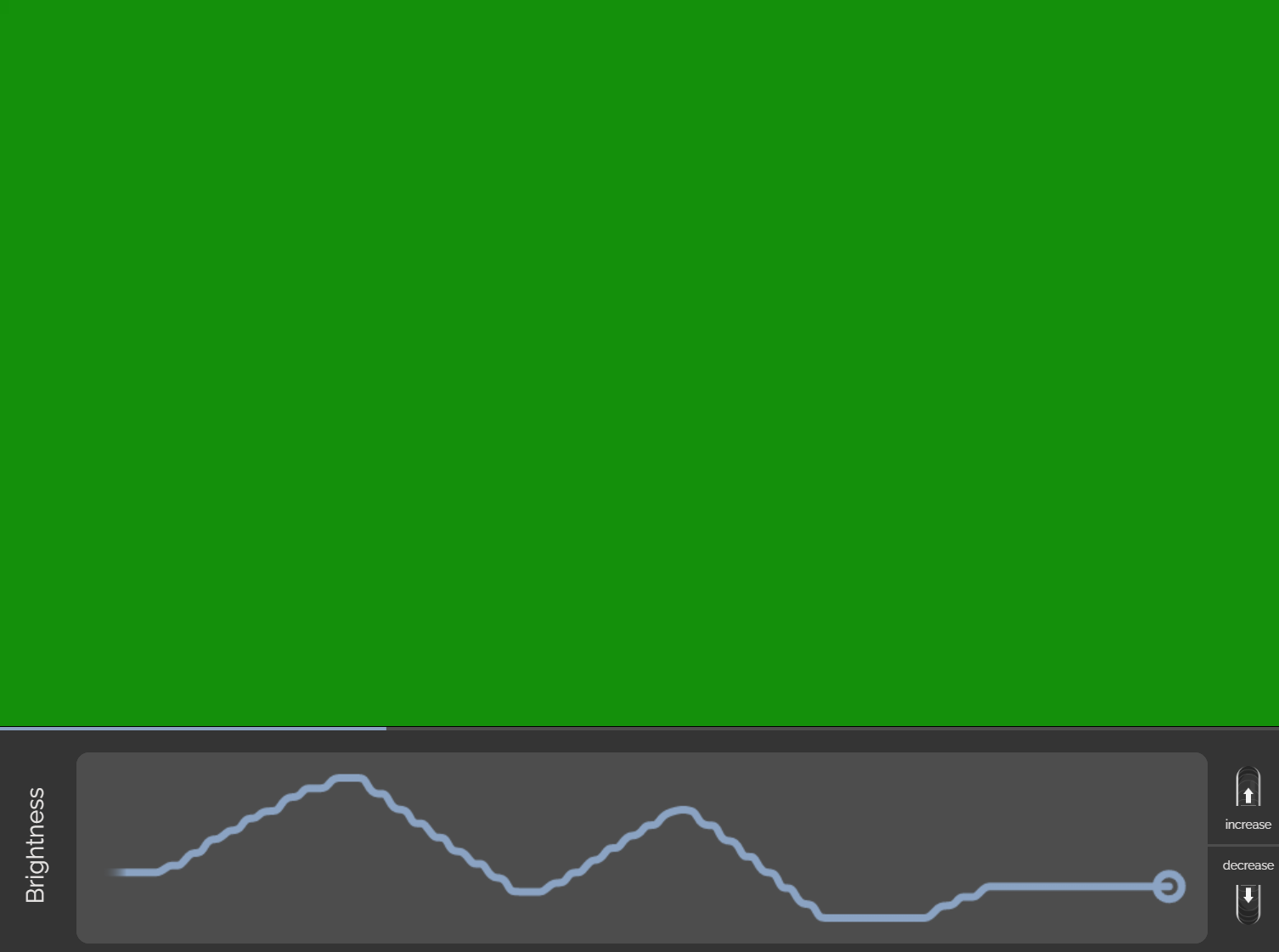}\\
{\footnotesize (a) Colour  task}
\end{minipage}\hfill
\begin{minipage}{0.3\columnwidth}\centering
\includegraphics[width=\linewidth]{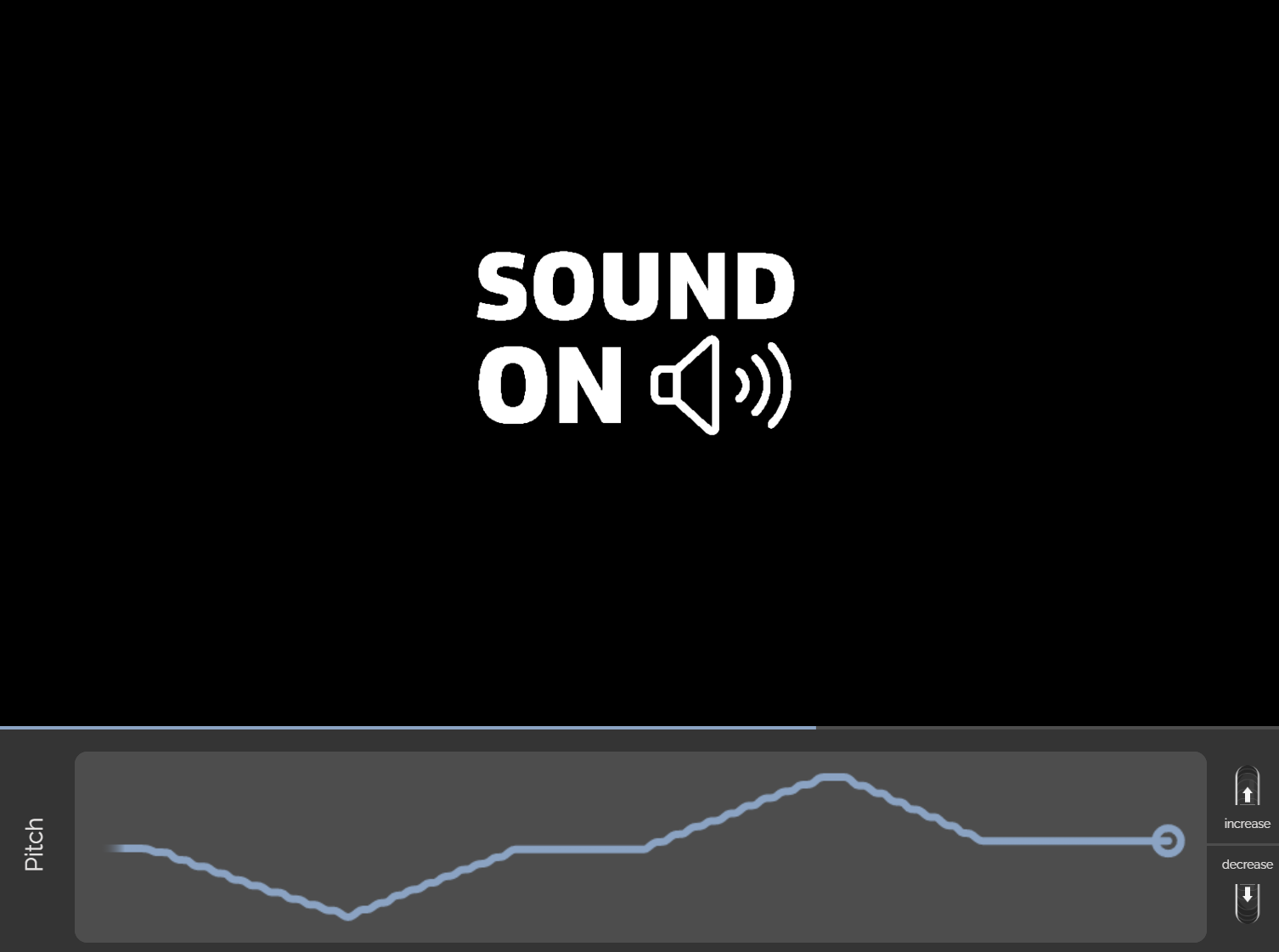}\\
{\footnotesize (b) Sound pitch task}
\end{minipage}\hfill
\begin{minipage}{0.37\columnwidth}\centering
\includegraphics[width=\linewidth]{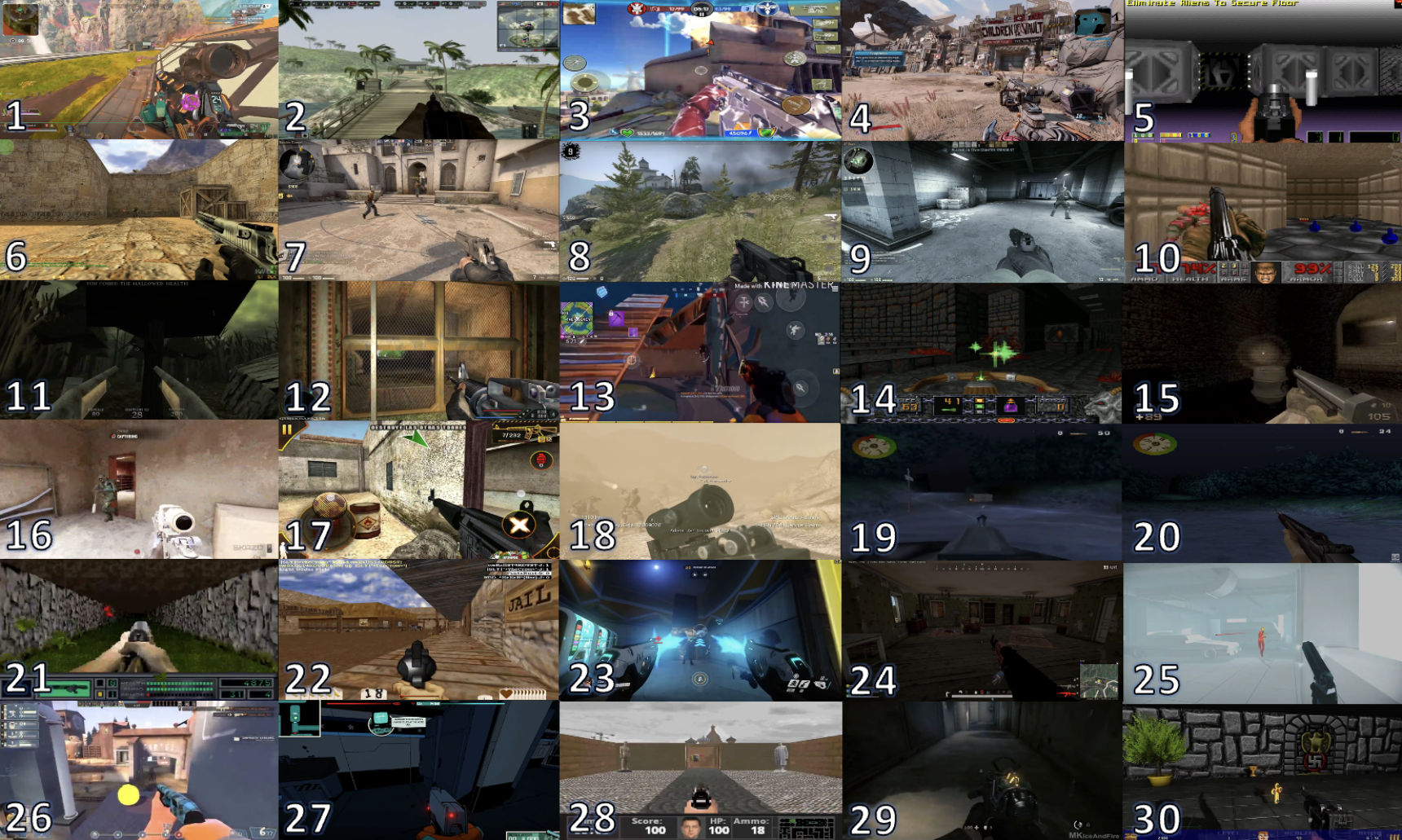}\\
{\footnotesize (c) Engagement task}
\end{minipage}\hfill

\caption{The PAGAN \cite{melhart2019pagan} annotation interfaces,
shared across both experiments: (a) colour intensity and (b) sound
pitch (Experiment~1, sensory), and (c) one-minute first-person shooter
clips for gameplay engagement (Experiment~2).}

\label{fig:tasks}
\end{figure}

\begin{figure}[t]
\centering
\begin{minipage}{0.75\columnwidth}\centering
\includegraphics[width=\linewidth]{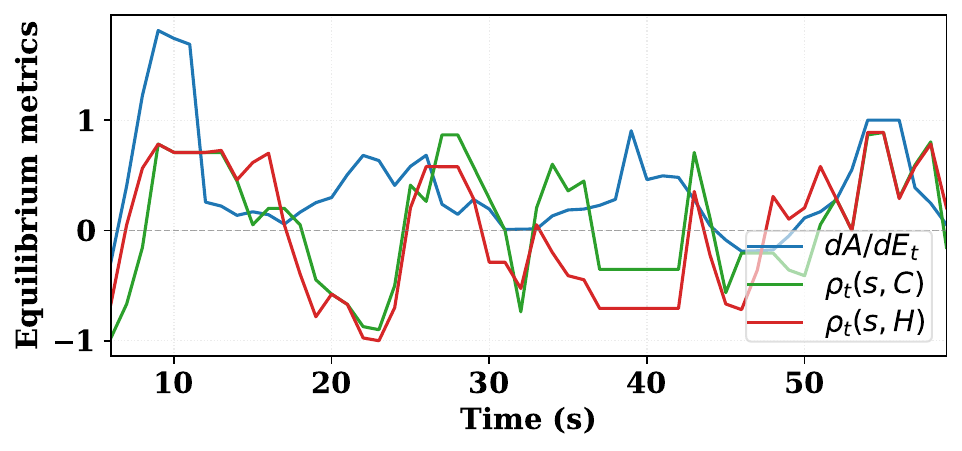}\\
{\footnotesize (a) Colour intensity task}
\end{minipage}
\begin{minipage}{0.75\columnwidth}\centering
\includegraphics[width=\linewidth]{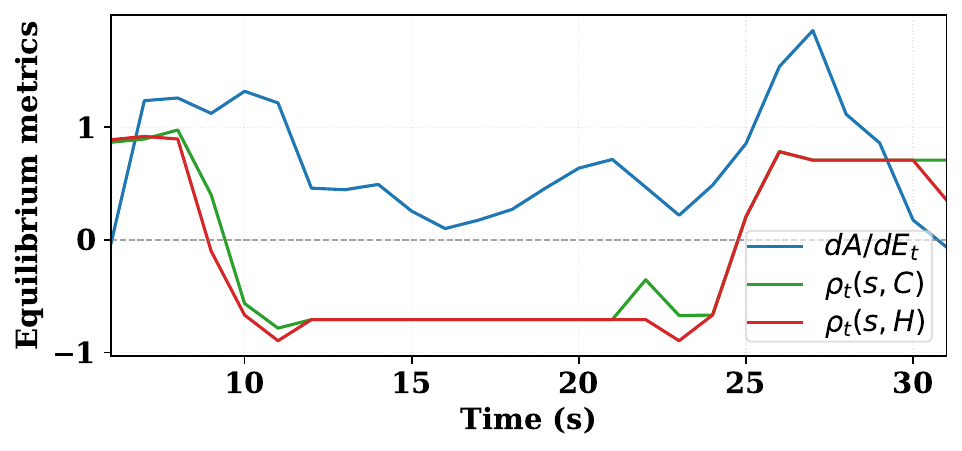}\\
{\footnotesize (b) Sound pitch task}
\end{minipage}
\caption{Temporal equilibrium paths for Experiment 1 (Sensory Tracking). The payoff--frontier slope $dA/dE_t$ (blue) remains robustly positive across the time horizon, validating the consensus-dominant prediction.}
\label{fig:controlled_dyn}
\end{figure}

\subsection{The Effort-Limited Regime ( Engagement Tracking)}
When the task shifts to gameplay engagement, the landscape inverts. Under the same uniform threshold ($\delta=0.05$) and framework, the system enters an effort-limited regime across all streams, the slope collapsing to a near-perfect negative correlation: Visual $-0.98$ $[-0.99,-0.96]$, Auditory $-0.99$ $[-0.99,-0.97]$, Audiovisual $-0.98$ $[-0.99,-0.95]$ (all saliency couplings $\rho(s,C),\rho(s,H)$ within $\pm0.10$ of zero). These three streams are not independent evidence---the audiovisual condition shares annotators and stimuli with its unimodal components---and the tight CIs abutting the $-1$ boundary are themselves consistent with a near-deterministic relationship, which is why the raw slope cannot by itself separate a behavioural inversion from a discretisation artefact.

\subsection{Ruling Out the Discretisation Artefact}\label{sec:residual}

\begin{figure}[t]
\centerline{\includegraphics[width=0.75\columnwidth]{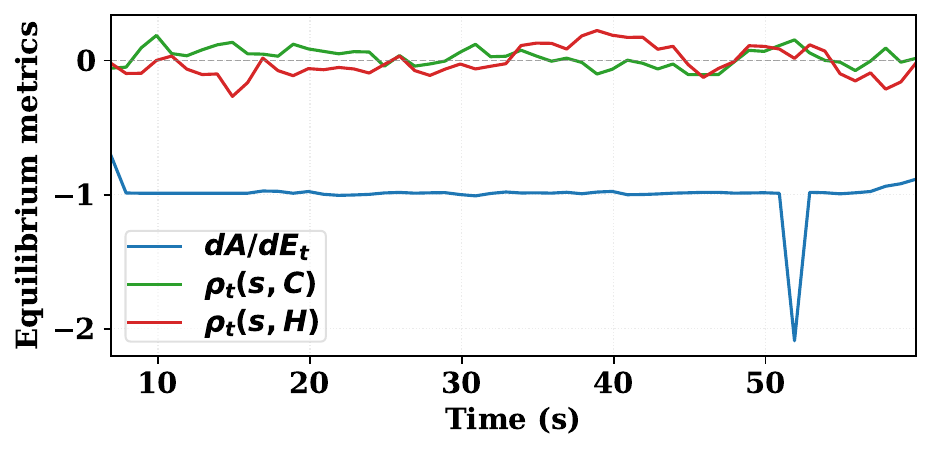}}
\caption{Windowed equilibrium curves for Experiment~2 (gameplay
engagement, $\delta=0.05$). The payoff--frontier slope $dA/dE_t$ (blue)
stays close to $-1$ throughout, contrasting with the positive sensory
frontier.}
\label{fig:gamevibe}
\end{figure}

$E_i$ and $A_i$ are derived from the same action stream $\{a_{i,t}\}$, so some negative coupling is inevitable: more active annotators are mechanically more likely to deviate from the instantaneous majority. Both experiments use the same interface, aggregation, and threshold ($\delta=0.05$); thus, this bias is identical across conditions, making the sensory task a built-in control. If the negative frontier were purely mechanical, both tasks would exhibit the same sign. Instead, sensory is strongly positive ($+0.63$, $+0.57$), whereas engagement is strongly negative ($-0.98$). Although cohorts and corpora also differ (Section~\ref{sec:datasets}), the identical pipeline between the two experiments makes construct abstraction the most parsimonious explanation. We therefore interpret the sign reversal—not the frontier magnitude—as the key result.
Figure~\ref{fig:gamevibe} shows the inversion persists across the $60$\,s session, with $dA/dE_t \approx -1$ for most windows. As engagement is semantically abstract, annotators follow private criteria $\nu_i(t)$ rather than a shared baseline, so idiosyncratic updates are rejected by majority vote, reducing coordination. The saliency couplings $\rho_t(s,C)$ and $\rho_t(s,H)$ remain near zero, while a sharp deviation near $52$\,s ($dA/dE_t \approx -2.1$) reflects conservative annotators holding steady as active ones diverged. Together with the sensory control, this confirms the negative frontier is a behavioural effect of semantic abstraction rather than a discretisation artefact.

\section{Discussion \& Conclusions}
This paper introduced the Ordinal Annotation Game, framing inter-annotator disagreement as a structured behavioural phenomenon shaped by construct abstraction. For concrete sensory judgments, annotators operate in a \emph{consensus-dominant} regime in which active updates reinforce alignment. In contrast, abstract appraisal tasks induce an \emph{effort-limited} regime, where frequent updates reduce rather than improve consensus. Since the identical annotation pipeline produces a \emph{positive} slope for the sensory task, the observed sign reversal reflects differences in the underlying construct rather than the discretisation procedure.

Several limitations bound these conclusions. First, the contrast is associational rather than causal, as the corpora differ in both complexity and participant population. Second, the small cohort ($N_2 = 5$) and the use of a low-level saliency proxy limit statistical power and generalisability. A more controlled test would manipulate construct abstraction within a single corpus while holding other factors constant. Future work could also incorporate psychophysical calibration of annotation schemas, evaluate alternative aggregation rules, and examine equilibrium behaviour in human--AI annotation settings.

Practically, the proposed framework provides a diagnostic for distinguishing structured from random disagreement \emph{before} label aggregation. It motivates construct-aware annotator weighting, treating active off-consensus labels as evidence of informed private appraisal rather than inattentiveness, and offers guidance for the design of annotation interfaces and task constructs. More broadly, it identifies settings in which a single majority-vote ground truth is inappropriate and a distributional representation of annotation is more faithful. By recasting disagreement as the outcome of an implicit equilibrium, the framework shifts its interpretation from measurement noise to an informative behavioural signal, motivating annotation and learning methods that explicitly account for the strategic complexity of perception.

\section*{Ethical Impact Statement}
This secondary re-analysis collects no new human-subjects data; the corpora were gathered under institutional ethics approvals with informed consent, GameVibe \cite{gamevibe} is public, and we use only anonymised, aggregated directional streams with no identifiable content. The per-annotator effort and bias statistics are analytical abstractions---not measures of competence---and must not be used to rank, surveil, or penalise annotators. The findings are associational and bounded by two corpora, small cohorts, and a low-level proxy, limiting generalizability; consistent with the ordinal view of affect, we treat annotations as behavioural measurements, not ground-truth readouts of felt emotion.

\section*{Acknowledgment}
© 2026 IEEE. Personal use of this material is permitted. Permission from IEEE must be obtained for all other uses, in any current or future media, including reprinting/republishing this material for advertising or promotional purposes, creating new collective works, for resale or redistribution to servers or lists, or reuse of any copyrighted component of this work in other works. 

This work has been accepted for publication at the 2026 International Conference on Affective Computing and Intelligent Interaction (ACII) Late-Breaking Research Track.

\bibliographystyle{IEEEtran}
\bibliography{refs}

@inproceedings{cowie2000feeltrace,
  title={Feeltrace: An instrument for recording perceived emotion in real time},
  author={Cowie, Roddy and Douglas-Cowie, Ellen and Savvidou, Susie and McMahon, Edelle and Sawey, Martin and Schr{\"o}der, Marc},
  booktitle={Proceedings of the ISCA Workshop on Speech and Emotion},
  volume={1},
  year={2000}
}

@inproceedings{melhart2019pagan,
  title={PAGAN: Video affect annotation made easy},
  author={Melhart, David and Liapis, Antonios and Yannakakis, Georgios N},
  booktitle={2019 8th international conference on affective computing and intelligent interaction (acii)},
  pages={130--136},
  year={2019},
  organization={IEEE}
}

@inproceedings{mariooryad2013analysis,
  title={Analysis and compensation of the reaction lag of evaluators in continuous emotional annotations},
  author={Mariooryad, Soroosh and Busso, Carlos},
  booktitle={2013 humaine association conference on affective computing and intelligent interaction},
  pages={85--90},
  year={2013},
  organization={IEEE}
}

@article{martinez2014don,
  title={Don’t classify ratings of affect; rank them!},
  author={Martinez, Hector P and Yannakakis, Georgios N and Hallam, John},
  journal={IEEE transactions on affective computing},
  volume={5},
  number={3},
  pages={314--326},
  year={2014},
  publisher={IEEE}
}

@article{cohen1960coefficient,
  title={A coefficient of agreement for nominal scales},
  author={Cohen, Jacob},
  journal={Educational and Psychological Measurement},
  volume={20},
  number={1},
  pages={37--46},
  year={1960},
  publisher={Sage Publications}
}

@inproceedings{ringeval2015av+,
  title={Av+ ec 2015: The first affect recognition challenge bridging across audio, video, and physiological data},
  author={Ringeval, Fabien and Schuller, Bj{\"o}rn and Valstar, Michel and Jaiswal, Shashank and Marchi, Erik and Lalanne, Denis and Cowie, Roddy and Pantic, Maja},
  booktitle={Proceedings of the 5th international workshop on audio/visual emotion challenge},
  pages={3--8},
  year={2015}
}

@article{dawid1979maximum,
  title={Maximum likelihood estimation of observer error-rates using the EM algorithm},
  author={Dawid, Alexander Philip and Skene, Allan M},
  journal={Journal of the Royal Statistical Society: Series C (Applied Statistics)},
  volume={28},
  number={1},
  pages={20--28},
  year={1979},
  publisher={Wiley Online Library}
}

@inproceedings{rodrigues2018deep,
  title={Deep learning from crowds},
  author={Rodrigues, Filipe and Pereira, Francisco},
  booktitle={Proceedings of the AAAI conference on artificial intelligence},
  volume={32},
  number={1},
  year={2018}
}

@article{geng2016label,
  title={Label distribution learning},
  author={Geng, Xin},
  journal={IEEE Transactions on Knowledge and Data Engineering},
  volume={28},
  number={7},
  pages={1734--1748},
  year={2016},
  publisher={IEEE}
}

@article{
hasson,
author = {Uri Hasson  and Yuval Nir  and Ifat Levy  and Galit Fuhrmann  and Rafael Malach },
title = {Intersubject Synchronization of Cortical Activity During Natural Vision},
journal = {Science},
volume = {303},
number = {5664},
pages = {1634-1640},
year = {2004},
doi = {10.1126/science.1089506},
URL = {https://www.science.org/doi/abs/10.1126/science.1089506},
eprint = {https://www.science.org/doi/pdf/10.1126/science.1089506}}

@article{wallot2018analyzing,
  title={Analyzing multivariate dynamics using cross-recurrence quantification analysis (crqa), diagonal-cross-recurrence profiles (dcrp), and multidimensional recurrence quantification analysis (mdrqa)--a tutorial in r},
  author={Wallot, Sebastian and Leonardi, Giuseppe},
  journal={Frontiers in psychology},
  volume={9},
  pages={2232},
  year={2018},
  publisher={Frontiers Media SA}
}

@article{palumbo2017interpersonal,
  title={Interpersonal autonomic physiology: A systematic review of the literature},
  author={Palumbo, Richard V and Marraccini, Marisa E and Weyandt, Lisa L and Wilder-Smith, Oliver and McGee, Heather A and Liu, Siwei and Goodwin, Matthew S},
  journal={Personality and social psychology review},
  volume={21},
  number={2},
  pages={99--141},
  year={2017},
  publisher={Sage Publications Sage CA: Los Angeles, CA}
}

@book{camerer2003behavioral,
  title={Behavioral game theory: Experiments in strategic interaction},
  author={Camerer, Colin},
  year={2003},
  publisher={Princeton university press}
}

@article{schelling1958strategy,
  title={The strategy of conflict. Prospectus for a reorientation of game theory},
  author={Schelling, Thomas C},
  journal={Journal of Conflict Resolution},
  volume={2},
  number={3},
  pages={203--264},
  year={1958},
  publisher={Sage Publications Sage CA: Thousand Oaks, CA}
}

@article{prelec2004bayesian,
  title={A Bayesian truth serum for subjective data},
  author={Prelec, Drazen},
  journal={science},
  volume={306},
  number={5695},
  pages={462--466},
  year={2004},
  publisher={American Association for the Advancement of Science}
}

@article{miller2005eliciting,
  title={Eliciting informative feedback: The peer-prediction method},
  author={Miller, Nolan and Resnick, Paul and Zeckhauser, Richard},
  journal={Management Science},
  volume={51},
  number={9},
  pages={1359--1373},
  year={2005},
  publisher={INFORMS}
}

@article{mohammad2013crowdsourcing,
  title={Crowdsourcing a word--emotion association lexicon},
  author={Mohammad, Saif M and Turney, Peter D},
  journal={Computational intelligence},
  volume={29},
  number={3},
  pages={436--465},
  year={2013},
  publisher={Wiley Online Library}
}

@article{gamevibe,
  title={GameVibe: a multimodal affective game corpus},
  author={Barthet, Matthew and Kaselimi, Maria and Pinitas, Kosmas and Makantasis, Konstantinos and Liapis, Antonios and Yannakakis, Georgios N},
  journal={Scientific Data},
  volume={11},
  number={1},
  pages={1306},
  year={2024},
  publisher={Nature Publishing Group UK London}
}

@inproceedings{lopes2017ranktrace,
  title={Ranktrace: Relative and unbounded affect annotation},
  author={Lopes, Phil and Yannakakis, Georgios N and Liapis, Antonios},
  booktitle={2017 Seventh International Conference on Affective Computing and Intelligent Interaction (ACII)},
  pages={158--163},
  year={2017},
  organization={IEEE}
}

@inproceedings{barthet2023knowing,
  title={Knowing your annotator: Rapidly testing the reliability of affect annotation},
  author={Barthet, Matthew and Trivedi, Chintan and Pinitas, Kosmas and Xylakis, Emmanouil and Makantasis, Konstantinos and Liapis, Antonios and Yannakakis, Georgios N},
  booktitle={2023 11th International Conference on Affective Computing and Intelligent Interaction Workshops and Demos (ACIIW)},
  pages={1--8},
  year={2023},
  organization={IEEE}
}

@inproceedings{pinitas2025privileged,
  title={Privileged Contrastive Pretraining for Multimodal Affect Modelling},
  author={Pinitas, Kosmas and Makantasis, Konstantinos and Yannakakis, Georgios},
  booktitle={Proceedings of the 27th International Conference on Multimodal Interaction},
  pages={317--325},
  year={2025}
}

@article{weidenholzer2010coordination,
  title={Coordination games and local interactions: a survey of the game theoretic literature},
  author={Weidenholzer, Simon},
  journal={Games},
  volume={1},
  number={4},
  pages={551--585},
  year={2010},
  publisher={MDPI}
}

@article{makantasis2023lab,
  title={From the lab to the wild: Affect modeling via privileged information},
  author={Makantasis, Konstantinos and Pinitas, Kosmas and Liapis, Antonios and Yannakakis, Georgios N},
  journal={IEEE Transactions on Affective Computing},
  volume={15},
  number={2},
  pages={380--392},
  year={2023},
  publisher={IEEE}
}

@inproceedings{pinitas2024varying,
  title={Varying the context to advance affect modelling: A study on game engagement prediction},
  author={Pinitas, Kosmas and Rasajski, Nemanja and Barthet, Matthew and Kaselimi, Maria and Makantasis, Konstantinos and Liapis, Antonios and Yannakakis, Georgios N},
  booktitle={2024 12th International Conference on Affective Computing and Intelligent Interaction (ACII)},
  pages={194--202},
  year={2024},
  organization={IEEE}
}

@article{itti2001computational,
  title={Computational modelling of visual attention},
  author={Itti, Laurent and Koch, Christof},
  journal={Nature reviews neuroscience},
  volume={2},
  number={3},
  pages={194--203},
  year={2001},
  publisher={Nature Publishing Group UK London}
}

@article{kayser2005mechanisms,
  title={Mechanisms for allocating auditory attention: an auditory saliency map},
  author={Kayser, Christoph and Petkov, Christopher I and Lippert, Michael and Logothetis, Nikos K},
  journal={Current biology},
  volume={15},
  number={21},
  pages={1943--1947},
  year={2005},
  publisher={Elsevier}
}

@inproceedings{pinitas2026lasca,
  title={LaScA: Language-Conditioned Scalable Modelling of Affective Dynamics},
  author={Pinitas, Kosmas and Maglogiannis, Ilias},
  booktitle={Proceedings of the IEEE/CVF Conference on Computer Vision and Pattern Recognition},
  pages={5379--5388},
  year={2026}
}

@inproceedings{pinitas2026beyond,
  title={Beyond the Mean: Modelling Annotation Distributions in Continuous Affect Prediction},
  author={Pinitas, Kosmas and Maglogiannis, Ilias},
  booktitle={Proceedings of the IEEE/CVF Conference on Computer Vision and Pattern Recognition},
  pages={5370--5378},
  year={2026}
}

\end{document}